\documentclass[a4paper,UKenglish,cleveref, autoref]{lipics-v2019}

\title{$\text{FO} = \text{FO}^3$ for linear orders with monotone binary relations} 

\author{Marie Fortin}{LSV, CNRS \& ENS Paris-Saclay, Universit{\'e} Paris-Saclay}{fortin@lsv.fr}{}{}

\authorrunning{M. Fortin}

\Copyright{Marie Fortin}



\nolinenumbers 

\hideLIPIcs  

\usepackage{amsmath,amssymb}
\usepackage{stmaryrd,multirow}
\usepackage{fancyhdr}
\usepackage{xspace}
\usepackage{mathtools}
\usepackage{cite}
\usepackage[dvipsnames]{xcolor}
\usepackage{tikz}
\usetikzlibrary{trees,arrows,shapes,snakes,fit,shadows,decorations.text,decorations.pathmorphing,decorations.pathreplacing,calc,automata,positioning,patterns}
\usepackage{enumitem}
\usepackage{todonotes}

\pgfdeclarelayer{background}
\pgfdeclarelayer{foreground}
\pgfsetlayers{background,main,foreground}

\newcommand\Nat{\mathbb{N}}
\newcommand{\df}{:=}

\newcommand{\Seta}{{A}}
\newcommand{\Setb}{{B}}
\newcommand{\Setc}{{C}}
\newcommand{\ela}{a}
\newcommand{\elb}{b}
\newcommand{\elc}{c}

\newcommand{\R}{\mathrel{\mathcal{R}}}
\newcommand{\comp}{\cdot}

\newcommand{\intpr}{interval-preserving\xspace}
\newcommand{\Intpr}{Interval-preserving\xspace}

\newcommand{\AP}{\mathcal{P}}
\newcommand{\propa}{P}
\newcommand{\propb}{Q}

\newcommand{\Rel}{\Gamma}
\newcommand{\rela}{\alpha}
\newcommand{\relb}{\beta}

\newcommand{\M}{\mathcal{M}}
\newcommand{\D}{A}

\newcommand{\leM}{\le}
\newcommand{\intM}[2]{{#2}^{#1}}

\newcommand{\FO}{\textup{FO}[\Gamma,\le]}
\newcommand{\FOD}{\textup{FO}[\Delta,\le]}
\newcommand{\FODt}{\textup{FO}^{3}[\Delta,\le]}
\newcommand{\FOk}[1]{\textup{FO}^{#1}[\Gamma,\le]}
\newcommand{\FOt}{\textup{FO}^3[\Gamma,\le]}
\newcommand{\Free}{\mathsf{Free}}
\newcommand{\Var}{\mathcal{X}}

\newcommand{\sfPDL}{\textup{PDL}_{\mathsf{sf}}[\Rel,\le]}
\newcommand{\sfPDLm}{\textup{PDL}_{\mathsf{sf}}[\Rel,\le,\Loopname,\mathsf{c1},\mathsf{c2},\mathsf{c3},\mathsf{c4}]}
\newcommand{\sfPDLmi}{\textup{PDL}_{\mathsf{sf}}[\Rel,\le,{\cap},\mathsf{c1},\mathsf{c2},\mathsf{c3},\mathsf{c4}]}
\newcommand{\existsp}[2]{\mathop{\langle #1 \rangle} {#2}}
\newcommand{\existsptrue}[1]{\mathop{\langle #1 \rangle}}
\newcommand{\True}{\mathit{true}}
\newcommand{\False}{\mathit{false}}
\newcommand{\compl}{\mathsf{c}}
\newcommand{\pic}[1]{{#1}^{\compl}}
\newcommand{\Loopname}{\mathsf{loop}}
\newcommand{\Loop}[1]{\Loopname(#1)}
\newcommand{\test}[1]{\{#1\}?}

\newcommand{\sem}[1]{\llbracket {#1} \rrbracket}
\newcommand{\sems}[1]{\left\llbracket {#1} \right\rrbracket}
\newcommand{\semM}[2]{{\llbracket {#2} \rrbracket}^{#1}}
\newcommand{\rightp}{<}
\newcommand{\righta}{\le}
\newcommand{\leftp}{>}
\newcommand{\lefta}{\ge}
\newcommand{\leftpi}[1]{{\mathsf{left}~} {#1}}

\newcommand{\rightpi}[1]{{\mathsf{right}~} {#1}}

\newcommand{\cll}[1]{{#1}^{\mathsf{c1}}}
\newcommand{\clr}[1]{{#1}^{\mathsf{c2}}}
\newcommand{\crl}[1]{{#1}^{\mathsf{c3}}}
\newcommand{\crr}[1]{{#1}^{\mathsf{c4}}}

\newcommand{\Until}{\mathbin{\mathsf{U}}}

\newcommand{\leftle}{\sqsubseteq_{\textup{left}}}
\newcommand{\rightle}{\sqsubseteq_{\textup{right}}}

\tikzstyle{dot}=[circle,fill,inner sep=0pt,minimum size=2pt]

\newcommand{\fok}[1]{\textup{FO}^{#1}}
\newcommand{\fo}{\textup{FO}}
\newcommand{\EF}{Ehrenfeucht-Fra\"iss\'e\xspace}
\newcommand{\C}{\mathcal{C}}

\newcommand{\Mg}{\M_{\Gamma}}
\newcommand{\Md}{\M_{\Delta}}
\newcommand{\equivp}{\equiv_{\Mg}}

\newcommand{\trad}[1]{\widetilde {#1}}

\definecolor{mygreen}{RGB}{0,200,0}

\begin{document}

\maketitle

\begin{abstract}
  We show that over the class of linear orders with additional binary
  relations satisfying some monotonicity conditions, monadic first-order logic
  has the three-variable property.
  This generalizes (and gives a new proof of) several known results, including
  the fact that monadic first-order logic has the three-variable property over linear
  orders, as well as over $(\mathbb{R},<,+1)$, and answers some open questions
  mentioned in a paper from Antonopoulos, Hunter, Raza and Worrell [FoSSaCS 2015].
  Our proof is based on a translation of monadic first-order logic formulas into
  formulas of a star-free variant of Propositional
  Dynamic Logic, which are in turn easily expressible in monadic first-order logic
  with three variables.
\end{abstract}

\section{Introduction}

Logics with a bounded number of variables have been extensively studied,
in particular in the context of descriptive complexity
\cite{Immerman82,Immerman91,Grohe98,KLPT06}
and temporal logics \cite{Kamp68,Gabbay1981,Hodkinson94,HunterOW13}.
One recurring question of interest
\cite{Gabbay1981,Poizat1982,ImmermanK89,Dawar05,Rossman08,AHRW15}
is to determine, in a given class $\C$ of structures, whether all properties
expressible in monadic first-order logic ($\fo$) can be defined in the fragment
$\fok k$ consisting of formulas which use at most $k$ variables. (A same
variable may be quantified over several times in a formula.)
In fact, several non-equivalent versions of this question appear in the
literature, many of which are compared in~\cite{HodkinsonS97}.
We say that $\C$ has the \emph{$k$-variable property} if every
formula of $\fo$ with at most $k$ free variables is equivalent
over $\C$ to a formula of $\fok k$.
Note that this is strictly stronger than requiring that all \emph{sentences}
(without free variables) of $\fo$ are equivalent to some $\fok k$ formulas.
Indeed, Hodkinson and Simon gave an example of a class of structures where
no sentence requires more than 3 variables, but which does not have the
$k$-variable property for any $k$ \cite{HodkinsonS97}.

The problem of whether a given class of structures has the $k$-variable property
is closely related to the question of the existence of an expressively
complete temporal logic (with a finite set of $\fo$-definable modalities).
A temporal logic is called expressively complete if any first-order formula
with a single free variable can be expressed in it.
For instance, it is well-known that linear temporal logic (LTL) over
Dedekind-complete time flows, or its extension with Stavi
connectives over all time flows, are expressively complete for
first-order logic \cite{Kamp68,GHR1993}.
More recently, it was shown that over the real numbers equipped with
binary relations $+q$ for all $q \in \mathbb{Q}$, metric temporal logic (MTL)
is expressively complete \cite{HunterOW13}.
However, the questions of having the $k$-variable property for some $k$
or admitting an expressively complete temporal logic are incomparable in
general: there exist a class of structures which admits a finite
expressively complete set of temporal connectives but which does not have
the $k$-variable property for any $k$ \cite{HodkinsonS97}, and one which
has the $3$-variable property but for which no temporal logic is
expressively complete \cite{Hodkinson94}.
However, Gabbay established that having the $k$-variable property
implies the existence of a \emph{multi-dimensional} expressively
complete temporal logic, with multiple reference points \cite{Gabbay1981}.

Another classical approach to proving or disproving that a class of structures
has the $k$-variable property is through \EF games, with a bounded number
of pebbles \cite{henkin1967,Poizat1982,ImmermanK89,AHRW15}.
This was applied by Immerman and Kozen to linear orders and bounded-degree
trees \cite{ImmermanK89}, and by Antonopoulos et al.\ to real-time
signals \cite{AHRW15}.

Natural candidates for classes $\C$ which might have the $k$-variable
property are classes of linearly ordered structures.
Indeed, a typical counter-example to unrestricted structures
having the $k$-variable property is a formula such as
``there exists $k+1$ distinct elements which satisfy some predicate $P$''.
It is in general not expressible in $\fok k$, but it is easily expressible
in $\fok 2$ if all models are equipped with a linear order $<$. For instance
for $k = 2$, we take the formula
$\exists x.\ P(x) \land
\exists y. (x < y \land P(y) \land \exists x. (y < x \land P(x)))$.
As mentioned before, Immerman and Kozen showed that the class of linear
orders has the $3$-variable property \cite{ImmermanK89}.
However, adding a single binary relation suffices to obtain a class of
linearly ordered structures which does not have the $k$-variable property
for any~$k$. Venema gave an example of a dense linear order with a single
equivalence relation which does not have the $k$-variable property for
any~$k$ \cite{Venema90}; this was adapted in \cite{AHRW15} to give
another example where the equivalence relation is replaced with a bijection.
In fact, even for finite linear orders, Rossman \cite{Rossman08}
proved that the class of linearly ordered graphs does not have the
$k$-variable property for any~$k$, resolving a problem which had been open
for more than 25 years \cite{Immerman82}.
Therefore, adding binary relations to linear orders while keeping the
$k$-variable property requires some restrictions on the interpretation
of the relation symbols.

On the positive side, Antonopoulos et al.\ proved that the class of structures
over $(\mathbb{R},<,+1)$ (or signals) has the $3$-variable property
\cite{AHRW15}.
Such structures have been studied in the context of real-time verification.
As a corollary, they also showed that $(\mathbb{R},<,f)$ has the
$3$-variable property  for any linear function $f : x \mapsto ax + b$.

\subparagraph{Contribution.}
We consider the class of linearly ordered structures with an additional
(finite or infinite) number of binary \emph{\intpr} relations.
These are binary relations $\R$ such that, for all intervals $I$, any point
which is in between two points of ${\R}(I)$ and has a preimage by $\R$
must have one in $I$. (We also require a symmetric condition of the converse
relation~${\R}^{-1}$.)
We show that $\fo$ over this class of structures also has the $3$-variable
property.

This generalizes results from \cite{ImmermanK89} and \cite{AHRW15}
described above.
Moreover, this answers some open questions mentioned in the conclusion
of \cite{AHRW15}, which asked if the result could be extended from linear
functions to polynomials over the reals, or other linear orders and families
of monotone functions.
In fact, all increasing or decreasing partial functions (over arbitrary
linear orders) are special cases of \intpr relations, and thus covered by
our result.

Our proof relies on different techniques than \cite{ImmermanK89,AHRW15},
which were based on \EF games.
We give an effective translation from $\fo$ to $\fok 3$ which goes through
a star-free variant of Propositional Dynamic Logic (PDL) with converse.
Propositional dynamic logic was introduced by Fischer and Ladner \cite{FisL79}
to reason about program schemes, and has now found a large range of
applications in artificial intelligence and verification
\cite{HalpernM92,GiacomoL94,LangeLutzJSL05,Lange06,Goeller2009}.
It combines local formulas containing modal operators, and path formulas
using the concatenation, union and Kleene star operations.
Several extensions have been studied, including PDL with converse
\cite{Streett81}, intersection \cite{Danecki85}, or negation of atomic
programs \cite{LutzW05}.
The particular \emph{star-free} variant of PDL we use here is in fact very
similar to Tarski's relation algebras \cite{TarskiGivant87}, which was used as
a basis for formalizing set theory.
It also corresponds to a \emph{two-dimensional} temporal logic in the
sense of Gabbay~\cite{Gabbay1981}.

We applied similar proof techniques in \cite{BolligFG18}, where we introduced
a star-free variant of PDL and proved that it is equivalent to $\fo$ over
\emph{message sequence charts} (MSCs) (and thus obtained a 3-variable property
result for MSCs as a corollary).
MSCs are discrete partial orders which represent behaviors of concurrent
message passing systems. They consist of a fixed, finite number of linear
orders called process orders (one for each process in the system), together with
FIFO binary message relations connecting matching send and receive actions.
Having a (fixed) finite number of total orders instead of a single one is not an
important difference, as we could always put them one after the other to
extend them into a single linear order.
FIFO relations are a special case of \intpr relations, thus the result
of the present paper can in fact be seen as a strict generalization of
our previous result in \cite{BolligFG18}.
More importantly, a major difference between MSCs studied in \cite{BolligFG18}
and the setting we consider here is that MSCs are \emph{discrete} structures,
whereas here we allow arbitrary linear orders. In fact, \cite{BolligFG18} relied
on the definition of formulas describing the minimum or the maximum of
some binary relations.
As such, it is interesting to see that the same kind of techniques can still
be applied to a priori very different linear orders.

\subparagraph{Outline.}
In Section \ref{sec:relations}, we introduce \intpr relations and monadic first-order
logic.
In Section \ref{sec:pdl}, we define star-free PDL, and prove some properties
of its formulas.
In Section \ref{sec:translation}, we give an effective translation from FO to
star-free PDL, and explain its consequences.
We conclude in Section \ref{sec:conclusion}.

\section{\Intpr relations and first-order logic}\label{sec:relations}

In this section, we define the class of structures covered by our results,
and recall the syntax of first-order logic.

\subparagraph{\Intpr binary relations.}
Let ${\R} \subseteq \Seta \times \Setb$ be a binary relation between sets
$\Seta$ and $\Setb$.
We write $\ela \R \elb$ if $(\ela,\elb) \in {\R}$, and
${\R}(\ela) = \{ \elb \in \Setb \mid \ela \R \elb\}$.
For a subset $\Seta' \subseteq \Seta$, we also write
${\R}(\Seta') = \bigcup_{\ela \in \Seta'} {\R}(\ela)$.
We define the \emph{converse} of a relation $\R$ as
${\R}^{-1} = \{(\elb,\ela) \in \Setb \times \Seta \mid (\ela,\elb) \in {\R}\}$,
and the \emph{composition} of two binary relations
${\R}_1 \subseteq \Seta \times \Setb$ and ${\R}_2 \subseteq \Setb \times \Setc$
as ${\R_1} \comp {\R_2} = \{(\ela,\elc) \in \Seta \times \Setc \mid \exists
\elb \in \Setb.\ (\ela,\elb) \in {\R_1} \land (\elb,\elc) \in {\R_2}\}$.
Finally, we write
${\R}^{\compl} = (\Seta \times \Setb) \setminus {\R}$ for the \emph{complement}
of $\R$.
Note that we have the following identities:
\[
  \left({\R}_1 \comp {\R_2}\right)^{-1} = {\R_2}^{-1} \comp {\R_1}^{-1}
  \qquad
  \left(\pic {\R}\right)^{-1} = \left( {\R}^{-1} \right)^{\compl}
  \qquad
  \left({\R}_1 \cap {\R_2}\right)^{-1} = {\R}_1^{-1} \cap {\R}_2^{-1} \, .
\]

A linear order ${\le}$ over a set $\Seta$ is a reflexive, transitive and
antisymmetric relation ${\le} \subseteq \Seta \times \Seta$ such that
for all $\ela,\elb \in \Seta$, we have $\ela \le \elb$ or $\elb \le \ela$.
Let $(\Seta,\le)$ be a linearly ordered set.
For $\Seta' \subseteq \Seta$, we will also denote by $\le$ the restriction of
$\le$ to $\Seta'$, so that $(\Seta',\le)$ is still a linearly ordered set.
Moreover, for $\ela \in \Seta$, we write $\ela < \Seta'$ if
for all $\ela' \in \Seta'$, $\ela < \ela'$,
and $\Seta' < \ela$ if for all $\ela' \in \Seta'$, $\ela' < \ela$.

An \emph{interval} of $(\Seta,\le)$ is a set $I \subseteq \Seta$ such that
for all $\ela \le \elb \le \elc$ with $\ela,\elc \in I$, we have $\elb \in I$.
For $\ela,\elb \in \Seta$, we denote by $[a,b)$ the interval
$\{c \in \Seta \mid a \le c < b\}$, and similarly for the intervals
$[a,b]$, $(a,b]$, $(a,b)$.
We call a relation ${\R} \subseteq \Seta \times \Setb$ between two linearly
ordered sets $(\Seta,\le_\Seta)$ and $(\Setb,\le_{\Setb})$ \emph{\intpr} if:
\begin{itemize}
\item For all intervals $I$ of $(\Seta,\le_\Seta)$, ${\R}(I)$ is an interval of
  $({\R}(\Seta),\le_\Setb)$.
\item For all intervals $J$ of $(\Setb,\le_\Setb)$, ${\R}^{-1}(J)$ is an
  interval of $({\R}^{-1}(\Setb),\le_\Seta)$.
\end{itemize}
In other terms, for all $\ela_1 \R \elb_1$ and $\ela_2 \R \elb_2$ with
$\ela_1,\ela_2 \in I$, for all $\elb_1 \le_\Setb \elb \le_\Setb \elb_2$,
if there exists some $\ela \in \Seta$ such that $\ela \R \elb$,
then there exists one in $I$
(cf.\ Figure~\ref{fig:def-intpr}).
Note that we do not require that \emph{all} elements between $\elb_1$ and
$\elb_2$ are in ${\R}(I)$, but only those which are in the image of ${\R}$.
The second condition is symmetric: for all $\ela_1 \R \elb_1$ and
$\ela_2 \R \elb_2$ with $\elb_1,\elb_2 \in J$, for all
$\ela_1 \le_\Seta \ela \le_\Seta \ela_2$,
if there exists some $\elb \in \Setb$ such that $\ela \R \elb$,
then there exists one in $J$.

\begin{figure}
  \begin{tikzpicture}[>=stealth,auto,yscale=1.4,inner sep=2pt]    
    \node (a1) at (0,1) {$a_1\vphantom{a'_2}$};
    \node (a2) at (2,1) {$a_2\vphantom{a'_2}$};
    \node (b1) at (0,0) {$b_1$};
    \node (b2) at (2,0) {$b_2$};
    \node (a) at (-1,1) {$a\vphantom{a'_2}$};
    \node (b) at (1,0) {$b$};
    \node (a') at (1,1) {$a'\vphantom{a'_2}$};

    \node at ($(b1)!0.5!(b)$) {$\le$};
    \node at ($(b)!0.5!(b2)$) {$\le$};

    \draw[->] (a1) -- (b1);
    \draw[->] (a2) -- (b2);
    \draw[->] (a) -- (b.north west);
    \path[->,dashed] (a') edge node {$\exists$} (b);

    \begin{pgfonlayer}{background}
      \fill[teal!20] (a1.north west) rectangle (a2.south east) ;
      \node[teal] at (2.5,1.1) {$I$};
    \end{pgfonlayer}
  \end{tikzpicture}
  \hspace{5em}
  \begin{tikzpicture}[>=stealth,auto,yscale=1.4,inner sep=2pt]    
    \node (a1) at (0,1) {$a_1$};
    \node (a2) at (2,1) {$a_2$};
    \node (b1) at (0,0) {$b_1$};
    \node (b2) at (2,0) {$b_2$};
    \node (b) at (-1,0) {$b$};
    \node (a) at (1,1) {$a$};
    \node (b') at (1,0) {$b'$};

    \node at ($(a1)!0.5!(a)$) {$\le$};
    \node at ($(a)!0.5!(a2)$) {$\le$};
    
    \draw[->] (a1) -- (b1);
    \draw[->] (a2) -- (b2);
    \draw[->] (a) -- (b.north east);
    \path[->,dashed] (a) edge node {$\exists$} (b');
    \begin{pgfonlayer}{background}
      \fill[violet!20] (b1.north west) rectangle (b2.south east) ;
      \node[violet] at (2.5,-0.1) {$J$};
    \end{pgfonlayer}
  \end{tikzpicture}
  \caption{Definition of \intpr relations.\label{fig:def-intpr}}
\end{figure}

\begin{example}\label{ex1}
  For any linear order $(\Seta,\le)$ and partial function
  $f : \Seta \to \Seta$, if $f$ is increasing or decreasing then
  the relation $\{(\ela,f(\ela)) \mid \ela \in \text{dom}(f)\}$
  is \intpr.

  As another example, consider a temporal structure $(\Seta,\le,\lambda)$
  over a set of atomic propositions $\text{AP}$, where
  $\lambda : \Seta \to 2^{\text{AP}}$ indicates the set of propositions
  which are true at a given point.
  For $\propa,\propb \in \text{AP}$, we let
  $\mathsf{until}_{\propa,\propb} =
  \{ (\ela,\elb) \in \Seta \times \Seta \mid
  \ela < \elb \land \propb \in \lambda(\elb) \land \forall \ela < \elc < \elb,
  \propa \in \lambda(\elc)\}$.
  Then $\mathsf{until}_{\propa,\propb}$ is \intpr.
\end{example}

The following lemma states some simple closure properties of \intpr relations.

\begin{lemma}\label{lem:closure}
  Let $(\Seta,\le_\Seta)$, $(\Setb,\le_\Setb)$, $(\Setc,\le_\Setc)$ be
  linearly ordered sets.
  \begin{enumerate}
  \item\label{closure-converse}
    For all \intpr relation ${\R} \subseteq \Seta \times \Setb$,
    ${\R}^{-1}$ is \intpr.
  \item\label{closure-inter}
    For all \intpr relations ${\R}_1, {\R}_2 \subseteq \Seta \times \Setb$,
    ${\R}_1 \cap {\R}_2$ is \intpr.
  \item\label{closure-comp}
    For all \intpr relations ${\R}_1 \subseteq \Seta \times \Setb$ and 
    ${\R}_2 \subseteq \Setb \times \Setc$,
    ${\R}_1 \comp {\R}_2$ is \intpr.
  \end{enumerate}
\end{lemma}

\begin{proof}
  Part \ref{closure-converse} follows from the fact that
  ${({\R}^{-1})}^{-1} = {\R}$.

  Let us prove \ref{closure-inter}.
  Since $({\R}_1 \cap {\R_2})^{-1} = {\R}_1^{-1} \cap {\R}_2^{-1}$,
  by symmetry, it suffices to prove that for all interval $I$ of $(\Seta,\le)$,
  $({\R}_1 \cap {\R}_2)(I)$ is an interval of
  $(({\R}_1 \cap {\R}_2)(\Seta),\le)$.
  Let $\ela_1,\ela_2 \in I$ and $\elb_1 \le \elb \le \elb_2$ such that
  $(\ela_1,\elb_1),(\ela_2,\elb_2) \in ({\R}_1 \cap {\R}_2)$ and
  $(\ela,\elb) \in ({\R}_1 \cap {\R}_2)$ for some $a \in \Seta$.
  If $\ela \in I$, then we are done.
  Otherwise, suppose for instance that $\ela < \ela_1 \le \ela_2$
  (the other cases are similar).
  Since $\R_1$ is \intpr, there exists $\ela_1 \le \ela' \le \ela_2$
  such that $\ela' \R_1 \elb$.
  Then, since $\ela < \ela_1 \le \ela'$ and ${\R}_1^{-1}(\elb)$ is an interval
  of $({\R}_1^{-1}(\Setb),\le_\Seta)$, we obtain $\ela_1 \R_1 \elb$.
  Similarly, $\ela_1 \R_2 \elb$.
  Hence $\ela_1 \mathrel{(\R_1 \cap \R_2)} \elb$.
  
  Let us show that \ref{closure-inter} implies \ref{closure-comp}.
  Again, by symmetry, it suffices to prove that for all interval $I$ of
  $(\Seta,\le_\Seta)$, $({\R_1} \cdot {\R_2})(I)$ is an interval of
  $(({\R_1} \cdot {\R_2})(A),\le_\Setc)$.
  Let ${\R_3} \subseteq \Setb \times \Setc$ denote the relation
  ${\R_1}(\Seta) \times \Setc$. It is an \intpr relation between
  $(\Setb,\le_\Setb)$ and $(\Setc,\le_\Setc)$. Moreover,
  we have $({\R_1} \cdot {\R_2})(A) = ({\R_2} \cap {\R_3})(B)$.
  Now, let $I$ be some interval of $(\Seta,\le_\Seta)$, and
  $J = \{\elb \in \Setb \mid \exists \elb_1,\elb_2 \in {\R_1}(I),
  \elb_1 \le \elb \le \elb_2\}$.
  Then $J$ is an interval of $(\Setb,\le_\Setb)$.
  Moreover, since $\R_1$ is \intpr, we have
  ${\R_1}(I) = J \cap {\R_1}(A)$, hence
  \[
    ({\R_1} \cdot {\R_2})(I) = {\R_2}({\R_1}(I)) = {\R_2}(J \cap {\R_1}(A))
    = ({\R}_2 \cap {\R}_3)(J) \, .
  \]
  Then, according to \ref{closure-inter},
  $({\R_1} \cdot {\R_2})(I)$ is an interval of
  $(({\R_2} \cap {\R_3})(B),\le_\Setc)$, i.e.,
  an interval of $(({\R_1} \cdot {\R_2})(A),\le_\Setc)$.
\end{proof}

\subparagraph{Models.}
Let $\AP = \{\propa,\propb,\ldots\}$ be an infinite set of monadic predicates,
and $\Rel = \{\rela,\relb,\ldots\}$ be a
finite or infinite set of binary relation symbols.
Throughout the paper, $\M$ will denote a structure $\M = (\D,\leM,
(\intM \M \rela)_{\rela \in \Rel},(\intM \M \propa)_{\propa \in \AP})$ where
$\leM$ is a linear order over $\D$,
$\intM \M \rela \subseteq \D \times \D$ is an \intpr relation for all
$\rela \in \Rel$,
and $\intM \M \propa \subseteq \D$ for all $\propa \in \AP$.

\subparagraph{Monadic first-order logic.}
We assume an infinite supply of variables $\Var = \{x,y,\ldots\}$.
The set $\FO$ of monadic first-order logic formulas over $\Rel$ is defined
as follows:
\[
  \Phi ::= \propa(x) \mid x \le y \mid x = y \mid \rela(x,y)
  \mid \Phi \lor \Phi \mid \lnot \Phi \mid \exists x. \Phi \, ,
  \qquad \text{where } x,y \in \Var, \propa \in \AP, {\rela} \in \Rel \, .
\]
We assume that all formulas are interpreted over structures $\M$ defined
as above.
Given an $\FO$ formula~$\Phi$, we denote by $\Free(\Phi)$ its set of
free variables.
We define the satisfaction relation $\M,\nu \models \Phi$ as usual, where
$\M = (\D,\leM,(\intM \M \rela)_{\rela \in \Rel},(\intM \M \propa)_{\propa \in \AP})$
and $\nu : \Free(\Phi) \to \D$ is an interpretation of the free
variables of $\Phi$.
We say that two formulas $\Phi,\Psi \in \FO$ are \emph{equivalent},
written $\Phi \equiv \Psi$, if for all $\M =
(\D,\leM,(\intM \M \rela)_{\rela \in \Rel},(\intM \M \propa)_{\propa \in \AP})$
and $\nu : \Free(\Phi) \cup \Free(\Psi) \to \D$, we have
$\M,\nu|_{\Free(\Phi)} \models \Phi$ if and only if
$\M,\nu|_{\Free(\Psi)} \models \Psi$.

For $k \in \Nat$, we denote by $\FOk k$ the set of first-order formulas with
at most $k$ variables. Note that a same variable may be quantified over
several times in the formula.

\begin{example}\label{ex:polynomials}
  Let $p : \mathbb{R} \to \mathbb{R}$ be a polynomial function, and
  $m_1<\cdots <m_n$ its local extrema (we suppose that $n \ge 1$).
  Fix $\Gamma = \{p\}$. For convenience, we will write $p(x) = y$
  instead of $p(x,y)$ in $\FO$ formulas.
  We focus on models of the form
  $\M = {(\mathbb{R},\le,\intM \M p, (\intM \M \propa)_{\propa \in \AP})}$
  where $\le$ is the usual ordering of the reals, and
  $\intM \M p = \{(x,p(x)) \mid x \in \mathbb{R}\}$.
  Let us describe an $\FOk 3$ formula $m_{i} \le x$  
  such that for all $\M$ and $r \in \mathbb{R}$, we have
  $\M,[x \mapsto r] \models m_i \le x$ if and only if $m_i \le r$.
  First, we write $p(x) \le p(y)$ for the $\FOk 3$ formula
  \[
    \exists z.\, p(x) = z \land \exists x.\, (p(y) = x \land z \le x) \, .
  \]
  We can then define formulas $\min (x) \in \FOk 3$ and
  $\max (x) \in \FOk 3$ which state that $x$ is a local minimum (resp.\ 
  maximum) of $p$, for instance:
  \begin{align*}
    \min (x) = {}
    & \left(\exists z.\, z < x \land \forall y.\,
      (z < y \le x \implies p(x) \le p(y))\right) \land {} \\
    & \left(\exists z.\, x < z \land \forall y.\,
      (x \le y < z \implies p(x) \le p(y))\right) \, .
  \end{align*}
  The formula $m_i \le x$ then states that there exists at least $i$ local
  extrema before $x$, alternating existential quantifications over $y$ and $z$
  to identify them; for instance, $m_3 \le x$ is the formula
  \[
    \exists y.\, y \le x \land (\min (y) \lor \max(y)) \land
    \exists z.\, z < y \land (\min (z) \lor \max(z)) \land
    \exists y.\ y < z \land (\min (y) \lor \max(y)) \, .
  \]
\end{example}

\section{Star-free Propositional Dynamic Logic}\label{sec:pdl}

\subparagraph{Star-free Propositional Dynamic Logic.}
Propositional dynamic logic (PDL) \cite{FisL79} consists of two sorts of formulas: state
formulas which are evaluated at single elements, and path formulas which are
evaluated at pairs of elements and allow to navigate inside the model.
Here we consider a star-free variant of PDL (with converse).
The syntax of \emph{star-free propositional dynamic logic} over
$\Rel$, written $\sfPDL$, is given below:
\begin{align*}
  \varphi & ::= \propa \mid \varphi \lor \varphi \mid \lnot \varphi \mid
            \existsp{\pi}{\varphi}
  && \text{(state formulas)} \\
  \pi & ::= {\rela} \mid {\righta} \mid \test{\varphi}
        \mid \pi^{-1}
        \mid \pi \cdot \pi
        \mid \pi \cup \pi \mid \pi \cap \pi \mid \pic \pi
  && \text{(path formulas)}
\end{align*}
where $\propa \in \AP$ and $\rela \in \Rel$.

Compared to classical PDL, star-free PDL uses the operators
$(\cdot,\cup,\cap,\compl)$ of star-free expressions, instead of the
rational operators $(\cdot,\cup,\ast)$.

Let $\M = (\D,\le,(\intM \M \rela)_{\rela \in \Rel},(\intM \M \propa)_{\propa \in \AP})$.
The semantics $\semM \M \varphi \subseteq \D$ or
$\semM \M \pi \subseteq \D \times \D$ of a state or path formula in $\sfPDL$
is defined below.
The state formula $\existsp{\pi}{\varphi}$ is true at a point $\ela \in \D$
in $\M$ (that is, $\ela \in \semM \M {\existsp{\pi}{\varphi}}$) if there exists some $\elb \in \D$ such that
$(\ela,\elb)$ satisfies $\pi$ and $\varphi$ is true at $\elb$.
The path formula $\test{\varphi}$ is stationary and tests if the state formula
$\varphi$ is true.
The semantics of other formulas is straightforward:
\begin{align*}
  \semM \M \propa & \df \intM \M \propa
  & \semM \M {\varphi_1 \lor \varphi_2}
  & \df \semM \M {\varphi_1} \cup \semM \M {\varphi_2} \\
  \semM \M {\lnot \varphi} & \df \D \setminus \semM \M \varphi
  & \semM \M {\existsp{\pi}{\varphi}}
  & \df \{ \ela \in \D \mid \exists \elb \in \semM \M \varphi,\
    (\ela,\elb) \in \semM \M \pi \} \\[2ex]
  \semM \M \rela & \df \intM \M \rela
  & \semM \M {\test{\varphi}} & \df \{(\ela,\ela) \mid \ela \in \semM \M \varphi\} \\
  \semM \M {\le} & \df {\leM}
  & \semM \M {\pi^{-1}} & \df {(\semM \M {\pi})}^{-1} \\
  \semM \M {\pi_1 \cup \pi_2} & \df {\semM \M {\pi_1}} \cup {\semM \M {\pi_2}} 
  & \semM \M {\pi_1 \cap \pi_2} & \df {\semM \M {\pi_1}} \cap {\semM \M {\pi_2}} \\
  \semM \M {\pic \pi} & \df (\D \times \D) \setminus \semM \M \pi
  & \semM \M {\pi_1 \cdot \pi_2} & \df {\semM \M {\pi_1}} \comp {\semM \M {\pi_2}} \, .
\end{align*}
For simplicity, we will often write $\sem \varphi$ or $\sem \pi$ instead of
$\semM \M \varphi$ and $\semM \M \pi$.
We also write $\M,\ela \models \varphi$ if $\ela \in \semM \M \varphi$, and
$\M,\ela,\elb \models \pi$ if $(\ela,\elb) \in \semM \M \pi$.

We will use the abbreviations $\True \df \propa \lor \lnot \propa$,
$\False \df \lnot \True$, ${\lefta} \df (\righta)^{-1}$,
${\rightp} \df \pic {\lefta}$, ${\leftp} \df \pic {\righta}$ and
$\existsptrue \pi \df \existsp \pi \True$.
For all $\sfPDL$ formulas $\pi$, we also define a state formula
$\Loop \pi \df \existsptrue {\pi \cap \test{\True}}$ which holds
at $\ela$ if and only if $(\ela,\ela) \in \sem \pi$.

\begin{example}
  Suppose that $\Rel = \{ +q \mid q \in \mathbb{Q} \}$, and that we consider
  only models over $\mathbb{R}$ and with $\sem {+q} = \{(r,r+q) \mid r \in \mathbb{R}\}$.
  Let $q,r \in \mathbb{Q}_{\ge 0}$ and $\propa,\propb \in \AP$.
  The formula $\propa \Until_{(q,r)} \propb$ of metric temporal logic,
  which holds at time $t \in \mathbb{R}$ if there exists $t+q < t' < t+r$ such that
  $t' \in \sem \propb$ and for all $t < t'' < t'$,  $t'' \in \sem \propa$,
  can be expressed in $\sfPDL$ as follows:
  \[
    \propa \Until_{(q,r)} \propb \equiv
    \left\langle{(+q \cdot{\rightp}) \cap (+r \cdot {\leftp}) \cap
    \pic {({\rightp} \cdot \test{\lnot \propa} \cdot {\rightp})}}\right\rangle
    {{\propb}}\, .
  \]
\end{example}

\subparagraph{An \intpr fragment of star-free PDL.}
We say that a path formula $\pi \in \sfPDL$ is \emph{\intpr} if for all $\M$,
$\semM \M \pi$ is \intpr.
Notice that $\le$ and $\test{\varphi}$ (for all $\varphi$) are \intpr.
By Lemma~\ref{lem:closure} (and assumption on $\sem \rela$), all $\sfPDL$
formulas constructed without the boolean operators $\cup$ and $\compl$ are
\intpr.
However, the complement or the union of \intpr relations are not in general
\intpr.
We define below a fragment of $\sfPDL$ where all path formulas are \intpr,
and which will turn out to be as expressive as $\sfPDL$ (and in fact, $\FO$)
when it comes to \emph{state} formulas.
To do so, we will introduce several restrictions of $\pic \pi$ which will
turn out to be \intpr, and which will suffice to characterize $\pic \pi$.

Let us first look at the different reasons for which we may have
$(\ela,\elb) \in \sem {\pic \pi}$, assuming that $\pi$ is \intpr.
To begin with, we focus on $\ela$.
One first sufficient condition for having $\elb \notin \sem \pi (\ela)$
is that $\sem \pi (\ela) = \emptyset$.
Now, suppose $\sem \pi (\ela) \neq \emptyset$.
If $\pi$ is \intpr, there are only three possible cases in which
$\elb \notin \sem \pi (\ela)$: $\elb < \sem \pi (\ela)$, or
$\sem \pi (\ela) < \elb$, or $\sem {\pi^{-1}} (\elb) = \emptyset$.
We define formulas $\leftpi \pi$ and $\rightpi \pi$ corresponding
respectively to the first two cases. We let
\begin{align*}
  \leftpi \pi
  & = \test{\existsptrue{\pi}} \cdot \pic{(\pi \cdot {\righta})} \, ,
    \quad \text{i.e.}
  & (\ela,\elb) \in \sem {\leftpi \pi}
  & \quad\text{iff}\quad
    b < \sem \pi (\ela) \neq \emptyset \\
  \rightpi \pi
  & = \test{\existsptrue{\pi}} \cdot \pic{(\pi \cdot {\lefta})} \, ,
    \quad \text{i.e.}
  & (\ela,\elb) \in \sem {\rightpi \pi}
  & \quad\text{iff}\quad
    b > \sem \pi (\ela) \neq \emptyset \, .
\end{align*}
Now, if we look at $\sem {\pi^{-1}}(\elb)$ instead of $\sem {\pi} (\ela)$,
we can make the same observations, by symmetry:
we have $(\ela,\elb) \in \sem {\pic \pi}$ if and only if
$\ela \notin \sem {\pi^{-1}}(\elb)$, and if $\pi$ is \intpr, there are again
only four possible cases:
$\sem {\pi^{-1}} (\elb) = \emptyset$,
or $\ela < \sem {\pi^{-1}}(\elb)$,
or $\ela > \sem {\pi^{-1}}(\elb)$,
or $\sem {\pi} (\ela) = \emptyset$.

Unfortunately, the formulas $\leftpi \pi$ and $\rightpi \pi$ are still
not \intpr in general.
However, if we take a more symmetric restriction of $\pic \pi$, where
we look at all the possible positions of $b$ and $a$ relatively to
$\sem \pi (\ela)$ and $\sem {\pi^{-1}} (\elb)$, we obtain four cases,
illustrated in Figure~\ref{fig:cll}, which we will later show correspond
to \intpr restrictions of $\pic \pi$.

More precisely, let
\begin{align*}
  \cll \pi
  & \df \leftpi \pi \cap {\left(\leftpi {\left(\pi^{-1}\right)}\right)}^{-1}
  \, , \text{i.e.}
  & (\ela,\elb) \in \sems {\cll \pi}
  & \quad\text{iff}\quad
    \begin{cases}
      b < \sems \pi (\ela) \neq \emptyset\\
      \ela < \sems {\pi^{-1}}(\elb) \neq \emptyset
    \end{cases}\\
  \clr \pi
  & \df \leftpi \pi \cap {\left(\rightpi {\left(\pi^{-1}\right)}\right)}^{-1}
  \, ,\text{i.e.}
  & (\ela,\elb) \in \sems {\clr \pi}
  & \quad\text{iff}\quad
    \begin{cases}
      b < \sems \pi (\ela) \neq \emptyset\\
      \ela > \sems {\pi^{-1}}(\elb) \neq \emptyset
    \end{cases}\\
  \crl \pi
  & \df \rightpi \pi  \cap {\left(\leftpi {\left(\pi^{-1}\right)}\right)}^{-1}
    \, ,\text{i.e.}
  & (\ela,\elb) \in \sems {\crl \pi}
  & \quad\text{iff}\quad
    \begin{cases}
      b > \sems \pi (\ela) \neq \emptyset\\
      \ela < \sems {\pi^{-1}}(\elb) \neq \emptyset
    \end{cases}\\
  \crr \pi
  & \df \rightpi \pi \cap {\left(\rightpi {\left(\pi^{-1}\right)}\right)}^{-1}
    \, ,\text{i.e.}
  & (\ela,\elb) \in \sems {\crr \pi}
  & \quad\text{iff}\quad
    \begin{cases}
      b > \sems \pi (\ela) \neq \emptyset\\
      \ela > \sems {\pi^{-1}}(\elb) \neq \emptyset \, .
    \end{cases}
\end{align*}
Notice that $\crl \pi \equiv {(\clr {(\pi^{-1})})}^{-1}$.

\begin{figure}
  \begin{tikzpicture}[auto,>=stealth,yscale=1.2,font=\small]
    \node[dot,label={[name=aa]above:{$\ela$}}] (a) at (0,1) {};
    \node[dot,label={[name=bb]below:{$\elb$}}] (b) at (0,0) {};
    \draw[fill,mygreen,opacity=0.3] (a) -- (0.7,0) -- (1.7,0) -- (a);
    \node[mygreen] (a') at (1.2,-0.25) {$\sem \pi (\ela)$};
    \draw[fill,violet,opacity=0.3] (b) -- (0.7,1) -- (1.7,1) -- (b);
    \node[violet] (b') at (1.2,1.25) {$\sem {\pi^{-1}}(\elb)$};
    \path[->] (a) edge node[left] {$\cll \pi$} (b);
    \node at ($(bb)!0.3!(a')$) {$<$};
    \node at ($(aa)!0.3!(b')$) {$<$};
    
    \begin{scope}[xshift=4.2cm]
      \node[dot,label={[name=aa2]above:{$\ela$}}] (a2) at (0,1) {};
      \node[dot,label={[name=bb2]below:{$\elb$}}] (b2) at (0,0) {};
      \draw[fill,mygreen,opacity=0.3] (a2) -- (0.7,0) -- (1.7,0) -- (a2);
      \node[mygreen] (a2') at (1.2,-0.25) {$\sem \pi (\ela)$};
      \draw[fill,violet,opacity=0.3] (b2) -- (-0.7,1) -- (-1.7,1) -- (b2);
      \node[violet] (b2') at (-1.2,1.25) {$\sem {\pi^{-1}}(\elb)$};
      \path[->] (a2) edge node[right,pos=0.6,inner sep=1pt] {$\clr \pi$} (b2);
      \node at ($(bb2)!0.3!(a2')$) {$<$};
      \node at ($(b2')!0.7!(aa2)$) {$<$};
    \end{scope}

    \begin{scope}[xshift=8.3cm]
      \node[dot,label={[name=aa3]above:{$\ela$}}] (a3) at (0,1) {};
      \node[dot,label={[name=bb3]below:{$\elb$}}] (b3) at (0,0) {};
      \draw[fill,mygreen,opacity=0.3] (a3) -- (-0.7,0) -- (-1.7,0) -- (a3);
      \node[mygreen] (a3') at (-1.2,-0.25) {$\sem \pi (\ela)$};
      \draw[fill,violet,opacity=0.3] (b3) -- (0.7,1) -- (1.7,1) -- (b3);
      \node[violet] (b3') at (1.2,1.25) {$\sem {\pi^{-1}}(\elb)$};
      \path[->] (a3) edge node[left,pos=0.6,inner sep=1pt] {$\crl \pi$} (b3);
      \node at ($(a3')!0.7!(bb3)$) {$<$};
      \node at ($(aa3)!0.3!(b3')$) {$<$};
    \end{scope}
    
    \begin{scope}[xshift=12.5cm]
      \node[dot,label={[name=aa4]above:{$\ela$}}] (a4) at (0,1) {};
      \node[dot,label={[name=bb4]below:{$\elb$}}] (b4) at (0,0) {};
      \draw[fill,mygreen,opacity=0.3] (a4) -- (-0.7,0) -- (-1.7,0) -- (a4);
      \node[mygreen] (a4') at (-1.2,-0.25) {$\sem \pi (\ela)$};
      \draw[fill,violet,opacity=0.3] (b4) -- (-0.7,1) -- (-1.7,1) -- (b4);
      \node[violet] (b4') at (-1.2,1.25) {$\sem {\pi^{-1}}(\elb)$};
      \path[->] (a4) edge node[right] {$\crr \pi$} (b4);
      \node at ($(a4')!0.7!(bb4)$) {$<$};
      \node at ($(b4')!0.7!(aa4)$) {$<$};
    \end{scope}
  \end{tikzpicture}
  \caption{Definition of $\cll \pi$, $\clr \pi$, $\crl \pi$ and $\crr \pi$,
    from left to right.\label{fig:cll}}
\end{figure}

Let $\sfPDLmi$ be the following restriction of $\sfPDL$:
\begin{align*}
  \varphi & ::= \propa \mid \varphi \lor \varphi \mid \lnot \varphi \mid
  \existsp{\pi}{\varphi} \\
  \pi & ::= {\rela} \mid {\le}
        \mid \test{\varphi} 
        \mid \pi^{-1} \mid \pi \cdot \pi \mid \pi \cap \pi
        \mid \cll \pi \mid \clr \pi \mid \crl \pi \mid \crr \pi \, .
\end{align*}

\begin{lemma}\label{lem:intpr}
  All $\sfPDLmi$ formulas are \intpr.
\end{lemma}

\begin{proof}
  We proceed by induction on the formula.
  By assumption, $\rela$ is \intpr for all $\rela \in \Rel$.
  Moreover, $\righta$ and $\test{\varphi}$ are \intpr.
  For $\pi^{-1}$, $\pi_1 \cdot \pi_2$ and $\pi_1 \cap \pi_2$,
  we apply Lemma~\ref{lem:closure}.

  Suppose that $\pi$ is \intpr.
  Let us show that $\cll \pi$ is \intpr.
  Notice that ${(\cll \pi)}^{-1} \equiv \cll {(\pi^{-1})}$.
  So we only need to show that for all intervals $I$, for all
  $\elb_1,\elb_2 \in \sem {\cll \pi} (I)$ and $\elb_1 \le \elb \le \elb_2$
  such that $\sem {(\cll \pi)^{-1}} (\elb) \neq \emptyset$, there exists
  $\ela \in I$ such that $(\ela,\elb) \in \sem {\cll \pi}$.
  Let $\ela_2 \in I$ such that $(\ela_2,\elb_2) \in \sem {\cll \pi}$.
  Let us show that we can in fact take $\ela = \ela_2$.
  The proof is illustrated in the picture below.
  
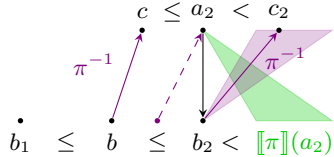
\begin{center}
    \begin{tikzpicture}[auto,>=stealth,yscale=1.2,font=\small]
      \node[dot,label={[name=aa2]above:{$\ela_2$}}] (a2) at (0,1) {};
      \node[dot,label={[name=bb2]below:{$\elb_2$}}] (b2) at (0,0) {};
      \node[dot,label={[name=bb]below:{$\elb$}}] (b) at (-1.2,0) {};
      \node[dot,label={[name=bb1]below:{$\elb_1$}}] (b1) at (-2.4,0) {};
      \draw[fill,mygreen,opacity=0.3] (a2) -- (0.7,0) -- (1.7,0) -- (a2);
      \node[mygreen] (a') at (1.2,-0.25) {$\sem \pi (\ela_2)$};
      \draw[fill,violet,opacity=0.2] (b2) -- (0.7,1) -- (1.7,1) -- (b2);
      \path[->] (a2) edge (b2);
      \node at ($(bb2)!0.3!(a')$) {$<$};
      \node[dot,label={[name=cc]above:{$c$}}] (c) at (-0.8,1) {};
      \node[dot,label={[name=cc2]above:{$c_2$}}] (c2) at (1,1) {};
      \node[violet,dot] (d) at (-0.6,0) {};
      
      \path[->,violet]
      (b) edge node[pos=0.4] {$\pi^{-1}$} (c)
      (b2) edge node[right,pos=0.7] {$\pi^{-1}$} (c2)
      (d) edge[dashed] (a2);

      \node at ($(bb)!0.5!(bb2)$) {$\le$};
      \node at ($(bb1)!0.5!(bb)$) {$\le$};
      \node at ($(cc)!0.5!(aa2)$) {$\le$};
      \node at ($(aa2)!0.5!(cc2)$) {$<$};
    \end{tikzpicture}
  \end{center}
  First, we have $b \le b_2 < \sem {\pi} (\ela_2) \neq \emptyset$.
  Moreover, $\sem {\pi^{-1}}(\elb) \neq \emptyset$
  (since $\sem {(\cll \pi)^{-1}} (\elb) \neq \emptyset$).
  Now, suppose towards a contradiction that
  $a_2 \not < \sem {\pi^{-1}}(\elb)$.
  Let $c \in \sem {\pi^{-1}}(\elb)$ such that $c \le a_2$.
  Since $(a_2,b_2) \in \sem {\cll \pi}$, there exists $c_2 > a_2$
  such that $(b_2,c_2) \in \sem {\pi^{-1}}$.
  We then have $c \le a_2 < c_2$ and $\sem {\pi} (a_2) \neq \emptyset$.
  Since $\pi$ is \intpr, we obtain $a_2 \in \sem {\pi^{-1}} ([b,b_2])$,
  a contradiction with the fact that $b_2 < \sem \pi (a_2)$.
  Thus, $(a_2,b) \in \sem {\cll \pi}$.

  Let us show that $\clr \pi$ is also \intpr.
  Similarly to the previous case, we show that for all
  $(\ela_2,\elb_2) \in \sem {\clr \pi}$
  and $\elb \le \elb_2$ such that $\sem {(\clr \pi)^{-1}}(b) \neq \emptyset$,
  we have $(\ela_2,\elb) \in \sem {\clr \pi}$.
  \begin{center}
    \begin{tikzpicture}[auto,>=stealth,yscale=1.2,font=\small]
      \node[dot,label={[name=aa2]above:{$\ela_2$}}] (a2) at (0,1) {};
      \node[dot,label={[name=bb2]below:{$\elb_2$}}] (b2) at (0,0) {};
      \node[dot,label={[name=bb]below:{$\elb$}}] (b) at (-1.2,0) {};
      \node[dot,label={[name=bb1]below:{$\elb_1$}}] (b1) at (-2.4,0) {};
      \draw[fill,mygreen,opacity=0.3] (a2) -- (0.7,0) -- (1.7,0) -- (a2);
      \node[mygreen] (a') at (1.2,-0.25) {$\sem \pi (\ela_2)$};
      \draw[fill,violet,opacity=0.2] (b2) -- (-0.7,1) -- (-1.7,1) -- (b2);
      \path[->] (a2) edge (b2);
      \node at ($(bb2)!0.3!(a')$) {$<$};
      \node[dot,label={[name=cc2]above:{$c_2$}}] (c2) at (-1.2,1) {};
      \node[dot,label={[name=cc]above:{$c$}}] (c) at (1,1) {};
      \node[violet,dot] (d) at (-0.4,0) {};
      
      \path[->,violet]
      (b) edge (c)
      (b2) edge (c2)
      (d) edge[dashed] (a2);

      \node at ($(bb)!0.5!(bb2)$) {$\le$};
      \node at ($(bb1)!0.5!(bb)$) {$\le$};
      \node at ($(cc)!0.5!(aa2)$) {$\le$};
      \node at ($(aa2)!0.5!(cc2)$) {$<$};
    \end{tikzpicture}
  \end{center}
  First, $\elb \le \elb_2 < \sem {\pi} (\ela_2) \neq \emptyset$, and
  $\sem {\pi^{-1}}(\elb) \neq \emptyset$.
  Suppose towards a contradiction that $\sem {\pi^{-1}} (\elb) \not < \ela_2$.
  Let $\elc \in \sem {\pi^{-1}} (b)$ such that $\ela_2 \le \elc$,
  and $\elc_2 \in \sem {\pi^{-1}}(\elb_2)$.
  We have $\elc_2 < \ela_2 \le \elc$, and $\sem \pi (\ela_2) \neq \emptyset$.
  Since $\pi$ is \intpr, we obtain $\ela_2 \in \sem {\pi^{-1}}([\elb,\elb_2])$,
  a contradiction with the fact that $\elb_2 < \sem {\pi} (\ela_2)$.
  Symmetrically, let $J$ be an interval,
  $\ela_1, \ela_2 \in \sem {(\clr \pi)^{-1}}(J)$, and
  $\ela_1 \le \ela \le \ela_2$ such that $\sem {\clr \pi} (a) \neq \emptyset$.
  Then for any $\elb_1 \in J$ such that $(\ela_1,\elb_1) \in \sem {\clr \pi}$,
  we also have $(\ela,\elb_1) \in \sem {\clr \pi}$, hence
  $a \in \sem {(\clr \pi)^{-1}}(J)$.

  Since $\crl \pi \equiv {(\clr {(\pi^{-1})})}^{-1}$, this also implies that
  $\crl \pi$ is \intpr.

  Finally, the case of $\crr \pi$ is symmetric to the case of $\cll \pi$:
  for all $(\ela_1,\elb_1) \in \sem {\crr \pi}$ and $\elb_1 \le \elb$
  such that $\sem {(\crr \pi)^{-1}} (\elb) \neq \emptyset$, we have
  $(\ela_1,\elb) \in \sem {\crr \pi}$.
\end{proof}

\section{Star-free PDL is expressively equivalent to FO}\label{sec:translation}

Let $\varphi$ be a state formula in $\sfPDL$, and $\Phi(x)$ an $\FO$ formula
with a single free variable $x$.
We say that $\varphi$ and $\Phi$ are equivalent, written
$\varphi \equiv \Phi(x)$, if for all $\M$ and elements $\ela$ in $\M$,
we have $\M,\ela \models \varphi$ if and only if
$\M,[x \mapsto \ela] \models \Phi(x)$.
Similarly, for a path formula $\pi \in \sfPDL$ and an $\FO$ formula
$\Phi(x,y)$ with exactly two free variables $x$ and $y$, we write
$\pi \equiv \Phi(x,y)$ if for all $\M$ and elements $\ela,\elb$ in $\M$,
we have $\M,\ela,\elb \models \pi$ if and only if
$\M,[x \mapsto \ela, y \mapsto \elb] \models \Phi(x,y)$.

\subparagraph{From $\boldsymbol{\sfPDL}$ to $\boldsymbol{\FOt}$.}
An easy induction shows that any formula in $\sfPDL$ can be translated
into an $\FO$ formula which uses at most three distinct variables:

\begin{lemma}\label{lem:PDL-to-FO}
  For every state formula $\varphi \in \sfPDL$, there exists a formula
  $\trad{\varphi}(x) \in \FOt$ such that
  $\varphi \equiv \trad{\varphi}(x)$.
  For every path formula $\pi \in \sfPDL$, there exists a formula
  $\trad{\pi}(x,y) \in \FOt$ such that $\pi \equiv \trad{\pi}(x,y)$.
\end{lemma}

\medskip

For the other direction, we will see that the fragment $\sfPDLm$ of $\sfPDL$
defined below will be sufficient:
\begin{align*}
  \varphi & ::= \propa \mid \varphi \lor \varphi \mid \lnot \varphi \mid
  \existsp{\pi}{\varphi} \mid \Loop \pi \\
  \pi & ::= {\rela} \mid {\le}
        \mid \test{\varphi} 
        \mid \pi^{-1} \mid \pi \cdot \pi
        \mid \cll \pi \mid \clr \pi \mid \crl \pi \mid \crr \pi \, .
\end{align*}
This fragment is a restriction of $\sfPDLmi$, where the intersection is only
used for $\Loop \pi$ formulas.

\subparagraph{From $\boldsymbol{\FO}$ to $\boldsymbol{\sfPDLm}$.}
The main result of the paper is an effective translation of $\FO$ formulas into
positive boolean combinations of formulas in $\sfPDLm$:

\begin{theorem}\label{thm:FO-to-PDL}
  Every formula $\Phi \in \FO$ with at least one free variable is equivalent
  to a positive boolean combination of formulas of the form $\trad \pi(x,y)$,
  where $x,y \in \Free(\Phi)$ and $\pi \in \sfPDLm$.
\end{theorem}

Note that the equivalent formula may also contain subformulas of the form
$\trad \pi(x,x)$.

\smallskip

Before proving Theorem~\ref{thm:FO-to-PDL}, we state some of its consequences.

\begin{corollary}
  Every formula $\Phi \in \FO$ with a single free variable is equivalent
  to some $\sfPDLm$ state formula.
  Every formula $\Phi \in \FO$ with two free variables is equivalent to some
  $\sfPDL$ path formula.
\end{corollary}

\begin{proof}
  If $\Phi$ has a single free variable $x$, it is equivalent to a positive
  boolean combination of formulas of the form $\trad \pi(x,x)$, which are
  themselves equivalent to the formulas $\Loop{\pi}$.
  The combination of these $\Loop \pi$ formulas is then a state formula
  of ${\sfPDLm}$.

  If $\Phi$ has two free variables $x$ and $y$, we obtain an equivalent
  positive boolean combination of formulas of the form $\trad \pi(x,y)$,
  $\trad \pi(y,x)$, $\trad \pi(x,x)$, or $\trad \pi (y,y)$.
  We can replace any subformula $\trad \pi (y,x)$ with $\trad {\pi^{-1}}(x,y)$,
  and any subformula $\trad \pi (x,x)$ with
  $\trad {\pi_1} (x,y) \lor \trad {\pi_2} (x,y)$,
  where $\pi_1 = (\test{\Loop \pi} \cdot {\le})$ and
  $\pi_2 = (\test{\Loop \pi} \cdot {\ge})$,
  and similarly for formulas $\trad \pi (y,y)$.
  We obtain an equivalent positive boolean combination of formulas of the form
  $\trad \pi(x,y)$.
  Since $\sfPDL$ allows union and intersection of path formulas, this is
  equivalent to a $\sfPDL$ formula.
\end{proof}

Another consequence is that $\FO$ over linear orders with \intpr relations
has the three-variable property. More precisely:

\begin{theorem}\label{thm:3-var}
  Any $\FO$ formula is equivalent to a boolean combination of formulas in
  $\FOk {3}$.
\end{theorem}

This also allows us to answer an open question from \cite{AHRW15}, namely,
whether structures over the real numbers with polynomial functions have
the $3$-variable property.
Suppose that $\Rel$ is a set of polynomials
$p : \mathbb{R} \to \mathbb{R}$.
Let $\Mg = (\mathbb{R},\le, (\intM \Mg p)_{p \in \Rel})$,
where $\le$ is the usual ordering of the real numbers,
and $\intM \Mg p = \{(x,p(x)) \mid x \in \mathbb{R}\}$ for all $p \in \Rel$.
Given an interpretation $h : \AP \to 2^{\mathbb{R}}$ of the monadic predicates,
we denote by $(\Mg,h)$ the structure
$(\mathbb{R},\leM, (\intM \Mg p)_{p \in \Rel},(h(\propa))_{\propa \in \AP})$.
We say that two formulas $\Phi, \Psi \in \FO$ are equivalent over $\Mg$,
written $\Phi \equivp \Psi$, if for all $h : \AP \to 2^{\mathbb{R}}$ and
$\nu : \Free(\Phi) \cup \Free(\Psi) \to \mathbb{R}$,
we have $(\Mg,h),\nu|_{\Free(\Phi)} \models \Phi$ if and only if
$(\Mg,h),\nu|_{\Free(\Psi)} \models \Psi$.

\begin{theorem}
  For all $\Phi \in \FO$, there exists a boolean combination $\Psi$
  of formulas in $\FOt$ such that $\Phi \equivp \Psi$.
\end{theorem}

\begin{proof}
  Let $p \in \Gamma$, and $m_1 < \cdots < m_n$ its local extrema.
  We denote by $p_{(-\infty,m_1)}$, $p_{[m_1,m_2)}, \ldots, p_{[m_n,+\infty)}$ the
  (monotone) restrictions of $p$ to intervals delimitated by its local extrema,
  and $\Delta_p$ the set of these partial functions.
  Let $\Delta = \bigcup_{p \in \Gamma} \Delta_p$.
  As above, we let
  $\Md = (\mathbb{R},\le, (\intM \Md {p_I})_{p_I \in \Delta})$,
  where $\le$ is the usual ordering of the real numbers, and
  $\intM \Md {p_I}= \{(x,p(x)) \mid x \in I\}$.
  Note that $\intM \Md {p_I}$ is \intpr (cf. Example~\ref{ex1}).
  We say that two formulas $\Phi \in \FO$ and $\Psi \in \FOD$ are equivalent,
  written $\Phi \equiv \Psi$, when for all $h : \AP \to 2^{\mathbb{R}}$ and
  $\nu : \Free(\Phi) \cup \Free(\Psi) \to \mathbb{R}$,
  we have $(\Mg,h),\nu|_{\Free(\Phi)} \models \Phi$ if and only if
  $(\Md,h),\nu|_{\Free(\Psi)} \models \Psi$.

  Let $\Phi \in \FO$.
  The formula $\Psi \in \FOD$ obtained by replacing each
  atomic formula $p(x,y)$ by $\bigvee_{p_I \in \Delta_p} p_I(x,y)$
  is equivalent to $\Phi$.
  Applying Theorem~\ref{thm:3-var} to $\Psi$, we obtain another
  formula $\Psi' \in \FOD$ such that $\Psi' \equiv \Psi$
  and $\Psi'$ is a boolean combination of formulas in $\FODt$.

  Following Example~\ref{ex:polynomials}, one can construct for
  each $p_I \in \Delta$ a formula ``$x \in I$'' of $\FOt$ such that
  $(\Mg,h), \nu \models x \in I$ if and only if $\nu(x) \in I$.
  Consider now the formula $\Phi' \in \FO$ obtained by replacing each
  atomic formula $p_I(x,y)$ in $\Psi'$ by $x \in I \land p(x,y)$.
  Then $\Phi' \equiv \Psi'$, hence $\Phi \equivp \Phi'$.
  Moreover, $\Phi'$ is a boolean combination of formulas in $\FOt$.
\end{proof}

The remainder of the section is devoted to the proof of
Theorem~\ref{thm:FO-to-PDL}.

\subparagraph{Eliminating negations.}
The fact that all $\sfPDLm$ path formulas are \intpr gives us a simple
characterization of the complement of a path formula: we show below that
an event $\elb$ is in $\sem {\pic \pi} (\ela)$ if it is to the left or to
the right of all elements of $\sem \pi (\ela)$, or if it does not satisfy
$\existsptrue{\pi^{-1}}$.
We can then show that the complement of a path formula in $\sfPDLm$ is
equivalent to a union of path formulas in $\sfPDLm$.
This will allow us to deal with negation in the translation from $\FO$
to $\sfPDLm$.

\begin{lemma}\label{lem:complement}
  For all path formulas $\pi \in \sfPDLm$, $\pic \pi$ is equivalent to
  a union of $\sfPDLm$ formulas.
\end{lemma}

\begin{proof}
  We show that
  \begin{align*}
    \pic \pi \equiv {}
    & (\test{\lnot \existsptrue{\pi}} \cdot {\righta})
    \cup
    (\test{\lnot \existsptrue{\pi}} \cdot {\lefta})
    \cup {} \\
    & ({\righta} \cdot \test{\lnot \existsptrue{\pi^{-1}}})
    \cup
    ({\lefta} \cdot \test{\lnot \existsptrue{\pi^{-1}}})
    \cup {} \\
    & (\cll \pi) \cup (\clr \pi) \cup (\crl \pi) \cup (\crr \pi)\, .
  \end{align*}
  We denote by $\pi'$ the right-hand-side formula.
  First, for all $\ela,\elb$ such that $\sem {\pi} (\ela) = \emptyset$
  or $\sem {\pi^{-1}} (\elb) = \emptyset$, we have
  $(\ela,\elb) \in \sem {\pic \pi}$ and $(\ela,\elb) \in \sem {\pi'}$.
  Now, suppose that $\sem {\pi} (\ela) \neq \emptyset$ and
  $\sem {\pi^{-1}} (\elb) \neq \emptyset$.
  We have $(\ela,\elb) \in \sem {\pi'}$ if and only if
  $(\ela,\elb) \in \sem {\cll \pi \cup \clr \pi \cup \crl \pi \cup \crr \pi}$.
  Clearly, if
  $(\ela,\elb) \in \sem {\cll \pi \cup \clr \pi \cup \crl \pi \cup \crr \pi}$,
  then $(\ela,\elb) \in \sem {\pic \pi}$.
  Conversely, let us show that if
  $(\ela,\elb) \notin \sem {\cll \pi \cup \clr \pi \cup \crl \pi \cup \crr \pi}$
  then $(\ela,\elb) \in \sem {\pi}$.
  In that case, we have either
  $\ela_1 \le \ela \le \ela_2$ for some
  $\ela_1,\ela_2 \in \sem {\pi^{-1}}(\elb)$, or
  $\elb_1 \le \elb \le \elb_2$ for some $\elb_1,\elb_2 \in \sem {\pi}(\ela)$.
  Since $\pi$ is \intpr, we obtain $(\ela,\elb) \in \sem \pi$.
\end{proof}

\subparagraph{Existential quantification.}
The elimination of existential quantifiers relies on the simple lemma below:

\begin{lemma}\label{lem:intersection-intervals}
  Let $(A,\le)$ be a linearly ordered set, and $I_1,\ldots,I_n$ intervals
  of $(A,\le)$. Then $\bigcap_{1 \le i \le n} I_i \neq \emptyset$ if and
  only if for all $1 \le i,j \le n$, $I_i \cap I_j \neq \emptyset$.
\end{lemma}

\begin{proof}
  We show that there exists $k$ and $\ell$ such that $\bigcap_{1 \le i \le n} I_i
  = I_k \cap I_\ell$, which implies the result.
  We define relations $\leftle$ and $\rightle$ over $\{I_1,\ldots,I_n\}$
  which, intuitively, compare respectively the left and right bounds of the
  intervals:
  \begin{align*}
    I \leftle J
    & \qquad\text{if}\qquad
      \forall a \in J, \exists b \in I, b \le a \\
    I \rightle J
    & \qquad\text{if}\qquad
      \forall a \in I, \exists b \in J, a \le b \, .
  \end{align*}
  It is easy to check that $\leftle$ and $\rightle$ are transitive, an that
  for all $I$ and $J$, we have $I \leftle J$ or $J \leftle I$ (or both),
  and similarly for $\rightle$.
  Thus, there exists $k$ such that $I_i \leftle I_k$ for all $i$, and
  $\ell$ such that $I_\ell \rightle I_i$ for all $i$.
  Then for all $a \in I_k \cap I_\ell$, for all $i$, there exists
  $b,b' \in I_i$ such that $b \le a \le b'$.
  Since $I_i$ is an interval, we obtain $a \in I_i$.
  Hence $I_k \cap I_\ell = \bigcap_{1 \le i \le n} I_i$.
\end{proof}

The next lemma follows from an application of
Lemma~\ref{lem:intersection-intervals} to intervals of the form
$\sem {\pi_i} (a_i)$.

\begin{lemma}\label{lem:exists}
  Let $n \ge 1$. For all path formulas $\pi_1,\ldots,\pi_n$ and all state
  formulas $\varphi$ in $\sfPDLm$, the $\FO$ formula
  \[
    \Phi = \exists x. \left(\trad \varphi(x) \land
    \bigwedge_{1 \le i \le n} \trad {\pi_i}(x_i,x)\right)
    \qquad (x_i \neq x \text{ for all $i$})
  \]
  is equivalent to a positive boolean combination of formulas of the form
  $\trad \pi(x_j,x_k)$, with $1 \le j,k \le n$ and $\pi \in \sfPDLm$.
\end{lemma}

\begin{proof}
  Let $\psi = \varphi \land \bigwedge_{1 \le i \le n} \existsptrue{\pi_i^{-1}}$,
  and
  \[
    \Psi = \bigwedge_{1 \le i,j \le n} \trad
    {(\pi_i \cdot \test{\psi} \cdot \pi_j^{-1})} (x_i,x_j) \, .
  \]
  Notice that $\Free(\Psi) = \Free(\Phi) = \{x_1,\ldots,x_n\}$.

  Let us show that $\Phi \equiv \Psi$.
  Let $\M = (\D,\le,(\intM M \rela)_{\rela \in \Rel},
  (\intM M \propa)_{\propa \in \AP})$,
  and $\nu : \{x_1,\ldots,x_n\} \to \D$.
  For all $1 \le i \le n$, let $I_i = \sem {\pi_i} (\nu(x_i)) \cap \sem \psi$.
  Let us show that $I_i$ is an interval of $(\sem \psi,\le)$.
  First, since $\pi_i$ is \intpr, $\sem {\pi_i} (\nu(x_i))$ is an interval
  of $(\sem {\existsptrue{\pi_i^{-1}}},\le)$.
  Thus, $I_i$ is an interval of
  $(\sem {\existsptrue{\pi_i^{-1}}} \cap \sem \psi,\le)$.
  But since $\sem {\existsptrue{\pi_i^{-1}}} \subseteq \sem \psi$,
  this is simply $(\sem \psi,\le)$.
  Besides, it is easy to verify that
  \[
    \M,\nu \models \Phi \qquad\iff\qquad
    \bigcap_{1 \le i \le n} I_i \neq \emptyset \, .
  \]
  Applying Lemma~\ref{lem:intersection-intervals}, we obtain
  \begin{align*}
    \M,\nu \models \Phi
    & \qquad\iff\qquad
    \text{for all $1 \le i,j \le n$},
      I_i \cap I_j
    \neq \emptyset \\
    & \qquad\iff\qquad
      \text{for all $1 \le i,j \le n$},
      (\nu(x_i),\nu(x_j)) \in \sem
      {\pi_i \cdot \test{\psi} \cdot {\pi_j^{-1}}} \\
    & \qquad\iff\qquad
      \M,\nu \models \Psi \, . \qedhere
  \end{align*}

\end{proof}

\subparagraph{Translation from $\boldsymbol{\FO}$ to $\boldsymbol{\sfPDLm}$.}
We are now ready to give the proof of Theorem~\ref{thm:FO-to-PDL}.

\begin{proof}[Proof of Theorem~\ref{thm:FO-to-PDL}]
  We assume that $\Phi$ is in prenex normal form, and prove the result
  by induction.
  The translation of atomic formulas $x \le y$ or $\rela(x,y)$ is
  straightforward; moreover, $\propa(x) \equiv \trad{\test{\propa}}(x,x)$, and
  $(x = y) \equiv \trad{\test{true}}(x,y)$.
  Using Lemma~\ref{lem:complement} to eliminate negations, we obtain the
  result for all quantifier-free formulas.

  The case $\Phi = \forall x. \Psi \equiv \lnot \exists x. \lnot \Psi$
  reduces to the case of existential quantification, applying again
  Lemma~\ref{lem:complement} to eliminate negations.

  We are left with the case $\Phi = \exists x. \Psi$. If $x$ is not free
  in $\Psi$, then $\Phi \equiv \Psi$ (since $\Psi$ has at least one free
  variable) and we are done by induction.
  Otherwise, assume that $\Free(\Psi)=\{x_1,\ldots,x_n\}$ with $n>1$ 
  and $x=x_n$.
  By induction, $\Psi$ is equivalent to a positive boolean combination 
  of formulas of the form $\trad \pi(x_i,x_j)$ with $\pi \in \sfPDLm$.
  We replace $\trad \pi(x_i,x_j)$ with $\trad{\pi^{-1}}(x_j,x_i)$ whenever $j < i$,
  and bring the resulting formula into disjunctive normal form.  
  Each conjunct is then of the form
  $\Upsilon = \Upsilon_1 \land \Upsilon_2 \land \Upsilon_3$, where
  $\Upsilon_1$ uses only variables from $\{x_1,\ldots,x_{n-1}\}$,
  $\Upsilon_2 = \bigwedge_{i} \trad{\pi_i}(y_i,x)$ with $y_i = x_j$ for some
  $1 \le j < n$,
  and $\Upsilon_3 = \bigwedge_{j} \trad{\pi_j}(x,x)$.
  Note that $\Upsilon_3  \equiv \trad{\varphi}(x)$, where
  $\varphi = \bigwedge_{j} \Loop {\pi_j}$.
  Then $\exists x. \Psi$ is equivalent to a finite disjunction of formulas
  \[
  \exists x. \Upsilon \; \equiv \;
  \Upsilon_1 \land \exists x. \left(\Upsilon_2 \land \trad\varphi(x)\right)
  \]
  with $\Upsilon_1$ and $\Upsilon_2$ as above.
  If $\Upsilon_2$ is empty, then we replace $\exists x. \trad\varphi(x)$ with
  the formula
  \[
  ({\righta}\cdot\test{\varphi}\cdot{\lefta})(x_1,x_1) \lor
  ({\lefta}\cdot\test{\varphi}\cdot{\righta})(x_1,x_1) \, .
  \]
  Otherwise, we apply Lemma~\ref{lem:exists} to
  $\exists x. \left(\Upsilon_2 \land \trad \varphi(x)\right)$.
  In all cases, we obtain an equivalent formula which is a positive boolean
  combination of formulas $\trad \pi(x_i,x_j)$ with $1 \le i,j < n$ and
  $\pi \in \sfPDLm$.
\end{proof}

\begin{remark}
  Without the assumption that all atomic binary relations are \intpr,
  $\sfPDL$ is still equivalent to $\FOt$.
  Indeed, in the proof of Theorem~\ref{thm:FO-to-PDL}, the assumption that all
  atomic binary relations are \intpr is only used in Lemmas~\ref{lem:complement}
  and~\ref{lem:exists}.
  But this assumption is not needed in the proof of Lemma~\ref{lem:exists}
  if $\Phi$ uses only three variables $x$, $y$ and~$z$. Indeed,
  we then have
  $\Phi \equiv \trad{(\pi \cdot \test{\varphi} \cdot \pi'^{-1})}(y,z)$,
  where $\pi$ is the
  intersection of all $\pi_i$ such that $x_i = y$, and $\pi'$ is the
  intersection of all $\pi_i$ such that $x_i = z$.
  Moreover, Lemma~\ref{lem:complement} is no longer needed if we translate
  an $\FOt$ formula into a positive boolean combination of $\sfPDL$ formulas,
  since $\sfPDL$ allows to take the complement of a path formula.
  Note that the equivalence with $\FOt$ is already proven in
  \cite{TarskiGivant87} (for the calculus of relations).
\end{remark}

\section{Conclusion}\label{sec:conclusion}

We proved that every $\FO$ formula over linear orders with \intpr binary
relations can be translated into an equivalent positive boolean combination of
path formulas in $\sfPDLm$.
In particular, any $\FO$ formula is equivalent to a boolean combination of
formulas in $\FOt$, which shows that $\FO$ has the three-variable property.
This generalizes several known results.

It would be interesting to see if the equivalence between $\FO$ and $\sfPDL$
can be used as an intermediate step to prove that a temporal logic is
expressively complete.
It is not the case in general, since \cite{HirshfeldR07} provides
an example of a class of structures which fits our assumptions but does not
admit any expressively complete temporal logic. However, the equivalence
could still be useful in more restricted settings.



\bibliography{lit.bib}

\end{document}